\documentclass[11pt]{article}
\usepackage[utf8]{inputenc}
\usepackage{amssymb}
\usepackage[longnamesfirst]{natbib}

\usepackage{enumitem}
\setlist[itemize]{nosep,topsep=3pt,itemsep=3pt}
\setlist[enumerate]{nosep,topsep=3pt,itemsep=3pt}
\usepackage{authblk}
\usepackage{xcolor}
\usepackage{comment}
\usepackage{multirow}
\usepackage[backref=page]{hyperref}

\newif\ifbackrefshowonlyfirst
\backrefshowonlyfirstfalse
\makeatletter
\let\BR@direct@old@hyper@natlinkstart\hyper@natlinkstart
\renewcommand*{\hyper@natlinkstart}{\phantomsection\BR@direct@old@hyper@natlinkstart}
\let\BR@direct@oldBR@citex\BR@citex
\renewcommand*{\BR@citex}{\phantomsection\BR@direct@oldBR@citex}%

\long\def\hyper@page@BR@direct@ref#1#2#3{\hyperlink{#3}{#1}}

\ifx\backrefxxx\hyper@page@backref
    \let\backrefxxx\hyper@page@BR@direct@ref
    \ifbackrefshowonlyfirst
    \fi
\else
    \ifbackrefshowonlyfirst
    \fi
\fi

\RequirePackage{etoolbox}
\patchcmd{\Hy@backout}{Doc-Start}{\@currentHref}{}{\errmessage{I can't seem to patch backref}}
\makeatother

\usepackage{etoolbox}
\makeatletter
\patchcmd{\NAT@citex}
  {\@citea\NAT@hyper@{%
     \NAT@nmfmt{\NAT@nm}%
     \hyper@natlinkbreak{\NAT@aysep\NAT@spacechar}{\@citeb\@extra@b@citeb}%
     \NAT@date}}
  {\@citea\NAT@nmfmt{\NAT@nm}%
   \NAT@aysep\NAT@spacechar\NAT@hyper@{\NAT@date}}{}{}

\patchcmd{\NAT@citex}
  {\@citea\NAT@hyper@{%
     \NAT@nmfmt{\NAT@nm}%
     \hyper@natlinkbreak{\NAT@spacechar\NAT@@open\if*#1*\else#1\NAT@spacechar\fi}%
       {\@citeb\@extra@b@citeb}%
     \NAT@date}}
  {\@citea\NAT@nmfmt{\NAT@nm}%
   \NAT@spacechar\NAT@@open\if*#1*\else#1\NAT@spacechar\fi\NAT@hyper@{\NAT@date}}
  {}{}
\makeatother


\usepackage{amsthm,amsfonts,framed}
\usepackage[margin=1in]{geometry}
\usepackage{physics, graphicx}

\usepackage[compact]{titlesec}
\titlespacing{\section}{0pt}{1.5ex}{0ex}
\titlespacing{\subsection}{0pt}{1.5ex}{0ex}
\titlespacing{\subsubsection}{0pt}{1ex}{0ex}
\titlespacing{\paragraph}{0pt}{1.5ex}{1ex}
\titleformat*{\paragraph}{\bfseries\itshape}

\usepackage{slantsc}
\theoremstyle{definition}

\theoremstyle{plain}
\newtheorem{thm}{Theorem} 

\theoremstyle{plain}

\usepackage{tikz,xcolor,hyperref}
\definecolor{lime}{HTML}{A6CE39}
\DeclareRobustCommand{\orcidicon}{
  \begin{tikzpicture}
  \draw[lime, fill=lime] (0,0)
  circle [radius=0.16]
  node[white] {{\fontfamily{qag}\selectfont \tiny ID}};
  \draw[white, fill=white] (-0.0625,0.095)
  circle [radius=0.007];
  \end{tikzpicture}
  \hspace{-2mm}
}
\foreach \x in {A, ..., Z}{\expandafter\xdef\csname orcid\x\endcsname{\noexpand\href{https://orcid.org/\csname orcidauthor\x\endcsname}
      {\noexpand\orcidicon}}
}

\hypersetup{
  colorlinks   = true,    
  urlcolor     = blue,    
  linkcolor    = blue,    
  citecolor    = blue 
}

\usepackage{xargs}
\usepackage[colorinlistoftodos,prependcaption,textsize=tiny]{todonotes}
\newcommandx{\jnote}[2][1=]{\todo[linecolor=red,backgroundcolor=red!25,bordercolor=red,#1]{J: #2}}
\newcommandx{\onote}[2][1=]{\todo[linecolor=green,backgroundcolor=green!25,bordercolor=green,#1]{O: #2}}
\newcommandx{\knote}[2][1=]{\todo[linecolor=orange,backgroundcolor=orange!25,bordercolor=orange,#1]{K: #2}}
\newcommandx{\improvement}[2][1=]{\todo[linecolor=Plum,backgroundcolor=Plum!25,bordercolor=Plum,#1]{#2}}
\newcommandx{\thiswillnotshow}[2][1=]{\todo[disable,#1]{#2}}

\begin{document}
\title{An explicit vector algorithm for high-girth MaxCut}
\author[1,2]{Jessica K. Thompson\thanks{jktho@umd.edu}}
\author[1]{Ojas Parekh\thanks{odparek@sandia.gov}}
\author[3]{\orcidA{}Kunal Marwaha\thanks{marwahaha@berkeley.edu}}
\affil[1]{Sandia National Laboratories, Albuquerque, NM, USA}
\affil[2]{Joint Center for Quantum Information and Computer Science, University of Maryland, College Park, MD}
\affil[3]{Berkeley Center for Quantum Information and Computation, University of California, Berkeley, CA}
\setcounter{Maxaffil}{0}
\renewcommand\Affilfont{\itshape\small}

\date{}
\maketitle

\thispagestyle{empty}
\setcounter{page}{0}
\begin{abstract}
We give an approximation algorithm for MaxCut and provide guarantees on the average fraction of edges cut on $d$-regular graphs of girth $\geq 2k$.
For every $d \geq 3$ and $k \geq 4$, our approximation guarantees are better than those of all other classical and quantum algorithms known to the authors.
Our algorithm constructs an explicit vector solution to the standard semidefinite relaxation of MaxCut and applies hyperplane rounding.
It may be viewed as a simplification of the previously best known technique, which approximates Gaussian wave processes on the infinite $d$-regular tree.
\end{abstract}
\newpage

\section{Introduction}
Finding the maximum cut in a graph is a canonical NP-hard discrete optimization problem, and one of the simplest constraint satisfaction problems.
However, to date, MaxCut does not have optimal, efficient approximation algorithms on most families of graphs.
Approximation algorithms are typically measured by one of two metrics: the \emph{performance} (or cut fraction), which is proportional to the cut size; or the \emph{approximation ratio}, which is the ratio of cut size and optimal cut size (or cut fraction and optimal cut fraction).

Arguably the most famous strategy to approximate MaxCut is to solve its semidefinite programming (SDP) relaxation and round the solution to get a cut \citep{goemans1995improved}.
Under the Unique Games Conjecture (UGC), this algorithm has the optimal \emph{worst-case} approximation ratio $\alpha > 0.878$ \citep{Khot02,KhotKMO07,odonnell2008}.
In other words, the cut given is of size at least $\alpha$ times the optimal cut size on \emph{every graph}.
However, many graphs, such as large random graphs, have an optimal cut fraction less than $1/(2\alpha)$ \citep{dembo2017extremal}.  In such cases, as an $\alpha$-approximation, Goemans-Williamson only guarantees a cut fraction less than $1/2$. In contrast, randomly assigning vertices to one of two partitions has an expected cut fraction of $1/2$. This is why considering the \emph{cut fraction} of an algorithm is relevant: for example, one can show that Goemans-Williamson has cut fraction $1/2$, just as random assignment. 

Even if UGC is true, we do not have approximation algorithms on \emph{regular} graphs that have an optimal worst-case approximation ratio \citep{halperin2004}. 
Regular graphs also admit straightforward bounds on cut fraction.
Many approximation algorithms (including ours) bound the maximum cut size as a function of basic graph parameters, including the graph degree and number of edges~\citep{erdos1967even,shearer1992note}. 
One example is a local threshold algorithm by \cite{hirvonen2017large}, which first randomly assigns vertices to one of two partitions; then, switches the label of each vertex if the number of cut edges adjacent to it does not meet a threshold.
Although the algorithm is simple, analyzing it is not straightforward; the optimal threshold is not known \emph{a priori}, so it must be numerically tuned for each degree value $d$ to maximize performance.

Quantum algorithms, especially the Quantum Approximate Optimization Algorithm (QAOA), may be competitive with classical algorithms on many optimization problems \citep{farhi2014quantum, farhi2015quantum}.
The QAOA includes a parameter $p$ corresponding to its \emph{depth}; its performance under optimal parameter settings, as well as its complexity, monotonically increases with $p$.
Although hardware limitations of quantum computers prevent the use of the QAOA at high depth \citep{Harrigan2021},
the QAOA may be competitive even at lower depths.
For example, when the QAOA was proposed, the depth-1 version had the best known cut fraction guarantee on triangle-free regular graphs \citep{Wang2018,ryananderson2018quantum}.
However, \cite{hastings2019classical} numerically showed that the local threshold algorithm \citep{hirvonen2017large} outperforms the depth-1 QAOA.
At higher $p$ and progressively higher girth, both algorithms perform better; however, a practitioner has to tune $O(p)$ parameters to perform optimally, and few provable performance bounds are known \citep{farhi2014quantum, Marwaha2021localclassicalmax}.

Another style of approximation algorithm uses invariant Gaussian processes \citep{csoka2015invariant, csoka2016independent, lyons2017factors, deboor2019, barak2021classical}.
These algorithms are not local, but they generally perform better than the threshold and quantum algorithms described previously.
First, one constructs an invariant Gaussian process satisfying an eigenvector equation on an infinite $d$-regular tree (more precisely, a Bethe lattice).
This \emph{Gaussian wave} process can then be approximated by a factor of an i.i.d.\ process assuming an eigenvalue condition is met.
One can use this to derive an independent set \citep{csoka2015invariant} or a cut \citep{csoka2016independent,lyons2017factors}.
As $d$ and the girth go to infinity, these algorithms approach a cut fraction of $1/2 + 2/(\pi\sqrt{d})$.
However, the analysis is quite technical; it involves approximating eigenvectors of the adjacency matrix of the infinite $d$-regular tree.
In contrast, our approach only requires the eigenvectors of the adjacency matrix of a finite path, which are well understood and explicitly known in certain cases.  In addition, the algorithms of \cite{csoka2015invariant} and \cite{lyons2017factors} were designed to prove existential results, and there is no claim of computational efficiency. The algorithm in \cite{barak2021classical} has a simpler form, but its analysis is complicated at finite degree; it is only exactly bounded in the large-degree limit.  The algorithm of \cite{deboor2019} is the closest in spirit to ours; however, ours offers better performance.

\paragraph{Our contributions.} We design an algorithm that retains desirable features of the above approaches but is simple to implement and analyze.  As with \cite{lyons2017factors}, our algorithm's performance bound is a function of the degree $d$ and girth $\geq 2k$, and we do not require either to be large.  In addition to a simplified approach, for $d \geq 3$ and $k \geq 4$, our approximation guarantees are better than those of all algorithms we are aware of.

We employ hyperplane rounding on a feasible SDP solution; however, instead of solving the SDP, we construct an explicit vector solution.  The vectors collectively depend on $k$ parameters, and we optimize the SDP objective subject to these parameters.  This optimization problem is considerably simpler than solving the SDP and amounts to finding a minimum eigenvector of the adjacency matrix of a weighted path on $k$ vertices.
Moreover, we give a closed-form approximate solution which also outperforms all previous algorithms.

Our approach is inspired by the work of \cite{carlson2020}, who construct and round explicit vector solutions to find large cuts in graphs with few triangles and in $K_r$-free graphs.  Furthermore, it turns out that our algorithm may be viewed as an simplification of the Gaussian wave process, as discussed in Section~\ref{sec:gaussian-wave}.
Its performance, as a function of $d$ when $k$ goes to infinity, matches that of other algorithms based on invariant Gaussian processes.

We motivate our discussion with a comparison to the Goemans-Williamson algorithm and then explicitly define our algorithm and establish a lower bound on its performance.  We then connect our algorithm to the Gaussian wave process, and prove that our bound is higher than several other performance guarantees.

\subsection{A motivating example}
Given a graph $G= (V,E)$, where $n=|V|$ and $m=|E|$, one can describe the unweighted\footnote{Our results readily generalize to the nonnegatively weighted case, replacing $m$ with the sum of all edge weights.}
MaxCut problem by the following binary quadratic program:
\begin{align*}
\text{max}\ \frac{1}{2}\ \sum_{(i,j)\in E} (1-z_i \cdot z_j) \ \text{such that } z_i \in \{-1,1\}, \ \forall i \in V
\end{align*}
This can be relaxed to the following semidefinite program (SDP), where $S^n$ is the unit sphere in $n$-dimensional space:
\begin{align*}
\text{max}\ \frac{1}{2} \sum_{(i,j)\in E} (1-\mathbf{v}_i \cdot \mathbf{v}_j) \ \text{such that } \mathbf{v}_i \in S^n, \ \forall i \in V
\end{align*}
\cite{goemans1995improved} show that given any feasible solution to the SDP (i.e., any set of unit vectors, $\{\mathbf{v}_i\}_{i=1}^n$ where $\mathbf{v}_i \in S^n$) and a random standard Gaussian vector $\mathbf{r}\in S^n$, the cut given by $A=\{i\ |\mathbf{v}_i\cdot \mathbf{r}>0\}$ and $\Bar{A}=V/A$ has a lower bound on its expected cut size (denoted by $E[W]$):
\begin{align*}
E[W] & =  \frac{1}{\pi} \sum_{(i,j)\in E}  \arccos(\mathbf{v}_i \cdot \mathbf{v}_j)           
       \geq \alpha \cdot \frac{1}{2} \sum_{(i,j)\in E} (1-\mathbf{v}_i \cdot \mathbf{v}_j)
\end{align*}
This scheme is known as \emph{hyperplane rounding}: the partition of vertex $i$ is set by the sign of $\mathbf{r} \cdot \mathbf{v}_i$.

Since the optimal solution to the relaxed program has a value at least the optimal cut size, this algorithm guarantees a cut with approximation ratio $\alpha$. The cut fraction is not as easy to bound because it depends on the values of the inner products $\mathbf{v}_i \cdot \mathbf{v}_j$. 

It turns out that the SDP relaxation can be used to give provably large cuts when the graph is locally tree-like.
We use the Goemans-Willamson rounding scheme; however, instead of actually solving the SDP, we explicitly construct a vector solution where $\mathbf{v}_i \cdot \mathbf{v}_j$ is same for all edges and then round to obtain a cut.
Consider the following example from \cite{carlson2020}, which is the inspiration for this paper.
On a $d$-regular, triangle-free graph, define for every $i,j \in V$:
\begin{equation*}
v_{ij} = \begin{cases}
\frac{1}{\sqrt{2}}   & i=j              \\
\frac{-1}{\sqrt{2d}} & (i,j)\in E       \\
0                    & \text{otherwise}
\end{cases}
\end{equation*}
Let $\mathbf{v}_i = [v_{i1}, ..., v_{in}]^T$.
Since $\mathbf{v}_i \in S^n$,  the set $\{\mathbf{v}_i\}_{i=1}^n$ is a feasible solution to the above semidefinite program. 
Note that $\mathbf{v}_i \cdot \mathbf{v}_j$ is the same for all edges $(i,j) \in E$.
Applying the Goemans-Williamson hyperplane rounding scheme gives a cut with expected size
\begin{align*}
E[W] & = \frac{1}{\pi} \sum_{(i,j)\in E} \arccos(\mathbf{v}_i \cdot \mathbf{v}_j) = \frac{m}{\pi}  \arccos(\frac{-1}{\sqrt{d}}).
\end{align*}
This algorithm outperforms depth-1 QAOA at every value of $d$  \citep{Wang2018}.
Although it is not as good as the threshold algorithm from \cite{hirvonen2017large}, we can generalize the vector assignment shown above for graphs with progressively higher girth.
On each of these graphs, the explicit vector solution has a performance guarantee above that of \cite{lyons2017factors}, the depth-1 and depth-2 QAOA, and  the threshold algorithm and its depth-2 extension in \cite{Marwaha2021localclassicalmax}.

\section{Vector assignments to the semidefinite program}
\label{sec:mainresult}

Consider a $d$-regular graph, $G= (V,E)$, where $n=|V|$ and $m=|E|$, with graph girth at least $2k$.
For each vertex, we define a vector $\mathbf{v}_i$ for all $i\in V$, where $\mathbf{v}_i = [v_{i1}, ..., v_{in}]^T$, and
\begin{equation*}
v_{ij} = \begin{cases}
\alpha_0    & i=j               \\
\alpha_\ell & dist(i,j) = \ell < k \\
0           & \text{otherwise}.
\end{cases}
\end{equation*}
Here, $dist(i,j)$ is the length of the shortest path between $i$ and $j$.
The algorithm is then to round the vectors $\mathbf{v}_i$ using the Goemans-Williamson rounding scheme.
We choose $\alpha_\ell$ to optimize the cut fraction. Notice that $\mathbf{v}_i \cdot \mathbf{v}_j$ is the same for all $(i,j)\in E$. Minimizing this value, where the vectors are restricted to being unit length, becomes the following optimization problem:
\begin{align*}
\min 2\alpha_0 \alpha_1 + 2(d-1)\alpha_1 \alpha_2 + \cdots + 2(d-1)^{k-2}\alpha_{k-2} \alpha_{k-1} \\ 
\text{such that } \alpha_0^2 + d\alpha_1^2 + d(d-1)\alpha_2^2 + \cdots + d(d-1)^{k-2}\alpha_{k-1}^2 = 1
\end{align*}
By scaling the $\alpha_\ell$'s, the problem becomes a minimization of a quadratic function over $S^k$:
\begin{gather*}
\min_{\|\boldsymbol{\beta}\|_2 = 1} \frac{2}{\sqrt{d}} \beta_0 \beta_1 + \frac{2\sqrt{d-1}}{d} \beta_1 \beta_2 + \cdots + \frac{2 \sqrt{d-1}}{d} \beta_{k-2} \beta_{k-1} = \min_{\|\boldsymbol{\beta}\|_2 = 1} \boldsymbol{\beta}^T A_k \boldsymbol{\beta} \\
\boldsymbol{\beta} = [\beta_0, ..., \beta_{k-1}]^T \\
\beta_\ell = \begin{cases}
\alpha_0                           & \ell = 0 \\
\alpha_\ell \sqrt{d(d-1)^{\ell-1}} & \ell > 0
\end{cases}
\end{gather*}
The minimum $\mathbf{v}_i \cdot \mathbf{v}_j$ is then the minimum eigenvalue of the following $k \times k$ matrix:
\[ A_k = \begin{bmatrix} 0 & a & 0      & 0      & 0      & \cdots \\
a & 0 & b      & 0      & 0      & \cdots \\
0 & b & 0      & b      & 0      & \cdots \\
&   & \ddots & \ddots & \ddots &        \\
0 & 0 & \cdots & b      & 0      & b      \\
0 & 0 & \cdots & 0      & b      & 0\end{bmatrix}
\quad a = \frac{1}{\sqrt{d}} \quad b = \frac{\sqrt{d-1}}{d}\]
The optimal $\beta_\ell$ are the entries of the corresponding eigenvector.
The $\alpha_\ell$ are then set from the $\beta_\ell$.

We now state our main result:
\begin{thm}
\label{thm:main}
Given a finite $d$-regular graph $G= (V,E)$ where $m = |E|$ with girth at least $2k$, there is an edge cut of $G$ with size
\begin{equation*}
 E[W] \geq \frac{m}{\pi} \arccos(\sigma_{d,k}),
\end{equation*}
such that 
$$\sigma_{d,k} \le \frac{-2\sqrt{d-1}}{d}
\Big( \cos  \frac{ \pi}{k+1}
+ \big(\sqrt{\frac{d}{d-1}} - 1\big) \frac{2}{k+1} \sin  \frac{\pi}{k + 1}  \sin  \frac{2 \pi}{k + 1} 
\Big)
<  \frac{-2\sqrt{d-1}}{d}
 \cos  \frac{ \pi}{k+1}\text{.}$$
\end{thm}
\begin{proof}
Choose parameters where $\boldsymbol{\beta} = \mathbf{w}$, such that
\begin{align*}
    w_{\ell} = \sqrt{\frac{2}{k+1}}\sin \frac{(\ell + 1) k \pi}{k+1} \quad
\mathbf{w}^T = [w_{0}, ..., w_{k-1}].
\end{align*}
Let's calculate the inner product $\mathbf{w}^T A_k \mathbf{w}$ in this case. Consider the following $k \times k$ matrix:
\begin{equation*}
B_k = \begin{bmatrix} 0 & b & 0      & 0      & 0      & \cdots \\
b & 0 & b      & 0      & 0      & \cdots \\
0 & b & 0      & b      & 0      & \cdots \\
&   & \ddots & \ddots & \ddots &        \\
0 & 0 & \cdots & b      & 0      & b      \\
0 & 0 & \cdots & 0      & b      & 0\end{bmatrix}
\end{equation*}
Since $B_k$ is a hollow tridiagonal Toeplitz matrix, its eigenvectors and eigenvalues are fully known \citep{noschese2013tridiagonal}. Here, $\mathbf{w}$ is the eigenvector associated with the minimum eigenvalue of $B_k$:
\begin{align*}
\mathbf{w}^T B_k \mathbf{w} = \lambda_{\min}(B_k) = 2b\cos \frac{k \pi}{k+1}
\end{align*}
Because of this, the following inequality holds:
\begin{align*}
\lambda_{\min}(A_k) \leq \mathbf{w}^TA_k\mathbf{w} = 2(a-b)w_0w_1 + \mathbf{w}^TB_k\mathbf{w} = 2(a-b)w_0w_1 + 2b\cos \frac{k \pi}{k+1}
\end{align*}
Since $\arccos$ is decreasing, the cut fraction is lower bounded:
\begin{align*}
E[W] = \frac{1}{\pi}\sum_{(i,j)\in E} \arccos(\mathbf{v}_i \cdot \mathbf{v}_j)
& = \frac{m}{\pi} \arccos(\lambda_{\min}(A_k)) \ge \frac{m}{\pi} \arccos(\mathbf{w}^T A_k \mathbf{w})
\end{align*}
Algebraic manipulation of the expression for $\mathbf{w}^T A_k \mathbf{w}$ finishes the proof.\footnote{Note that because $w_0 w_1 \le 0$ and $a \ge b$, $\mathbf{w}^T A_k \mathbf{w} \le \mathbf{w}^T B_k \mathbf{w} = \lambda_{min}(B_k)$.}
\end{proof}
Assigning $\boldsymbol{\beta}$ as the minimum eigenvector of $B_k$ does not require any additional calculation, because a closed-form expression for the eigenvector is known for all $k$. However, this algorithm has highest performance when assigning $\beta_\ell$ from the minimum eigenvector of $A_k$.
At large $d$ or $k$, the difference between $A_k$ and $B_k$ tends to zero.

When $d \geq 3$ and $k \geq 4$, this is the best known cut fraction for $d$-regular graphs of girth $\geq 2k$.
See Section \ref{sec:bounds} for a comparison of this bound with other algorithms' performance guarantees.

\section{Relationship to the Gaussian wave process}
\label{sec:gaussian-wave}

Consider a vector assignment given by our algorithm, $\{\mathbf{v}_i\}_{i=1}^n$, and let $V = [\mathbf{v}_1 ... \mathbf{v}_n]$.
Hyperplane rounding can now be thought of as sampling from a $n$-variable Gaussian distribution $\mathbf{x} \sim \mathcal{N}(0^n, \Sigma)$ (where $\Sigma = V^T V$), and taking the sign of each variable.
Any multivariate Gaussian distribution can be expressed as linear combinations of i.i.d.\ Gaussians.  In other words and when applied to our setting, $\mathbf{x} = V \mathbf{z}$, where $\mathbf{z}$ encodes $n$ i.i.d\ Gaussian random variables.  This yields
\begin{equation*}
\mathbf{x}_i = \sum_{j \in V} \alpha_{d(i,j)} \mathbf{z}_j = \sum_{\ell = 0}^{k-1} \sum_{j:d(i,j) = \ell} \alpha_\ell \mathbf{z}_j.
\end{equation*}

Thus, our algorithm may be viewed as a linear factor of an i.i.d\ process, presented in \cite{csoka2015invariant} to approximate the Gaussian wave and used in \cite{lyons2017factors}:
\begin{equation*}
\mathbf{x}_i = \sum_{j \in V} \alpha_{d(i,j)} \mathbf{z}_j = \sum_{\ell = 0}^\infty \sum_{j:d(i,j) =\ell} \alpha_\ell \mathbf{z}_j
\end{equation*}
The first equation may be recovered by setting $\alpha_\ell = 0$ for all $\ell \geq k$.

\cite{csoka2015invariant} construct a linear factor of an i.i.d.\ process on the infinite $d$-regular tree to approximate the Gaussian wave process.
They observe that in this setting, one may assume that only a finite number of the $\alpha_\ell$ are nonzero; however, their results are on graphs of sufficiently high girth, and they do not explicitly consider girth as a parameter. 

The approach of Chapter 4.1 in
\cite{deboor2019} is similar to ours.
The main difference is that de Boor uses one free parameter $c$ to determine $\{\alpha_\ell\}_{l=0}^{k-1}$ rather than obtaining the optimal values as we do.
In fact, de Boor's optimal choice of $c$ exactly reproduces the block FIID in Theorem 4.2 of \cite{lyons2017factors}, with the same bound on cut fraction.\footnote{de Boor makes a mathematical mistake; on page 20, his $\zeta^2$ should equal $1 + c^2 d L$.}
Our approach is guaranteed to perform as well any other approach based on a linear factor of an i.i.d.\ process using only $\{\alpha_\ell\}_{\ell=0}^{k-1}$.


\section{Comparison of performance guarantees}
\label{sec:bounds}
The expected cut size given by our algorithm is larger than the guarantee given by Theorem 4.2 of \cite{lyons2017factors}.
We state this formally:
\begin{thm}
\label{thm:comparison}
Given a finite $d$-regular graph $G= (V,E)$ with girth at least $2k$, the explicit vector solution has a higher performance guarantee than that of the algorithm in \cite{lyons2017factors} for every $k \geq 3$ and $d \geq 3$.
\end{thm}
\begin{proof}
First consider this form of the expected cut fraction:
\begin{equation*} 
E[W]= \frac{m}{\pi} \arccos(-2 b \cdot \xi)
\end{equation*}
We compare $\xi$ for the explicit vector solution and the algorithm in \cite{lyons2017factors}.
Call $\xi$ the \emph{relative expectation}.
The explicit vector solution has relative expectation
\begin{align*}
\xi_{EV} & 
\geq \frac{\lambda_{\min}(B_k)}{-2 b}  = \cos \frac{\pi}{k+1}   \geq 1 - 0.5 (\frac{\pi}{k+1})^2.
\end{align*}
The algorithm in \cite{lyons2017factors} has relative expectation
\begin{align*}
    \xi_{L} = (1 + \frac{d-1}{d(k-1)})^{-1} = \frac{k-1}{k-1/d}.
\end{align*}
The condition $\xi_{EV} > \xi_{L}$ holds when
\begin{align*}
    & &  1 - 0.5(\frac{\pi}{k+1})^2 &> \frac{k-1}{k-1/d} \\
    &\Leftrightarrow & k - 1/d &> \frac{k-1}{1 -  0.5(\frac{\pi}{k+1})^2} \\
    &\Leftrightarrow &  d &> \Big(k - (\frac{k-1}{1 -  0.5(\frac{\pi}{k+1})^2}) \Big)^{-1}.
\end{align*}
The righthand expression decreases with $k$, so if the condition holds at $k_*$, it holds at all $k \geq k_*$.

Consider $k\in \{3,4,5\}$.
The condition holds when
\begin{align*}
    d &> 9.26 & &(k=3) \\
    d &> 3.82 & &(k=4) \\
    d &> 2.75 & &(k=5).
\end{align*}
All that remains to show is that the relative expectation is higher for ($k=3$, $3 \leq d \leq 9$) and ($k=4$, $d=3$).
A numerical calculation verifies this; see Table \ref{table:computedvals}.
This completes the proof.
\end{proof}

The explicit vector solution also outperforms the QAOA and threshold algorithm, at depth-1 and depth-2, for all $k \geq 3$ and $d \geq 3$.
Our solution has performance increasing with $k$, and the other algorithms have performance increasing with depth, so we compare the $k=3$ solution to the depth-2 algorithms.
Consider the cut fraction $1/2 + c_d/\sqrt{d}$, where $c_d$ is the \emph{normalized value}.
At large $d$, our algorithm has normalized value $c_\infty = \sqrt{2}/\pi > 0.45$, higher than that of the QAOA ($c_\infty \approx 0.41$) and the threshold algorithm ($c_\infty\approx 0.42)$.
The normalized value decreases with $d$ for $d \geq 3$; for example, $c_8 \approx 0.454$.
Compare this with Figure 2 and Table 1 of \cite{Marwaha2021localclassicalmax}; the explicit vector solution outperforms the QAOA and threshold algorithm at all $d \geq 3$.

The modified threshold algorithm of $d \in \{3,4,5\}$ was found with a computer search over a large class of algorithms \citep{Marwaha2021localclassicalmax}.
This outperforms our solution at $k=3$ but not at $k=4$.

\begin{table}
\begin{center}
\begin{tabular}{|c|c|c|c|} \hline
k                  & d & Explicit vector & Lyons bound \\ \hline
\multirow{7}{*}{3} & 3 & 0.78656                      & 0.75000        \\ \cline{2-4}
                   & 4 & 0.76180                      & 0.72727         \\ \cline{2-4}
                   & 5 & 0.74883                      & 0.71428         \\ \cline{2-4}
                   & 6 & 0.74085                      & 0.70588         \\ \cline{2-4}
                   & 7 & 0.73543                      & 0.70000        \\ \cline{2-4}
                   & 8 & 0.73151                      & 0.69565          \\ \cline{2-4}
                   & 9 & 0.72855                    & 0.69230        \\ \hline
\multirow{1}{*}{4} & 3 & 0.85927                    & 0.81818         \\ \hline
\end{tabular}
\end{center}
\caption{\small Comparison of relative expectation of our solution and the lower bound of a previous Gaussian wave algorithm.
The calculation uses the bound $\mathbf{w}^T A_k \mathbf{w}$ given in the proof of Theorem \ref{thm:main}, where $\mathbf{w}$ is the minimum eigenvector of $B_k$.
All values are truncated.}
\label{table:computedvals}
\end{table}

\section{Conclusions}
We provide a procedure to find a large cut on a high-girth graph.
For a given girth $g \ge 8$, this algorithm provides a cut fraction guarantee higher than that of all other classical and quantum algorithms to date.
It gives a simple, explicit solution to the semidefinite program.  Moreover, the process of using hyperplane rounding on this solution can be thought of as a simplification of the Gaussian wave process previously used.

Although our analysis requires $k$ upper bounded by the girth, in practice one can use $k$ equal to the diameter of the graph.
Numerically, this change improves the cut fraction, but a new analysis is needed to prove it.

One can also use the explicit vector technique in other optimization problems.
For example, it can be directly applied to the independent set algorithm in \cite{csoka2015invariant}.
Though this does not improve the performance, further work (for example a tighter bound) could lead to an improvement.

There is recent interest in the development of local improvement algorithms which start with an existing approximation (as opposed to random initialization).
For example, \cite{Egger_2021} implement ``warm-start'' QAOA, initializing the quantum computer with a solution produced by the Goemans-Williamson algorithm.
This technique could also be applied here.
Using the output of our algorithm, one could apply a local threshold procedure to improve the cut.
Although previous results are numerical, analytic bounds may be possible here because of our solution's simple form.

No current algorithm provably achieves the optimal cut fraction on locally tree-like graphs \citep{dembo2017extremal}.
It is suspected that one may exist using approximate message-passing \citep{alaoui2020optimization, montanari2021optimization}; but for now, this is an open question.
This means that better algorithms than the one we propose may exist, including better quantum algorithms.

\section*{Acknowledgements}
KM thanks Boaz Barak for personal guidance. OP and JT are affiliated with Sandia National Laboratories, which is a multimission laboratory managed and operated by National Technology and Engineering Solutions of Sandia, LLC., a wholly owned subsidiary of Honeywell International, Inc., for the U.S. Department of Energy’s National Nuclear Security Administration under contract DE-NA-0003525. OP and JT were supported by the U.S. Department of Energy, Office of Science, Office of Advanced Scientific Computing Research, Accelerated Research in Quantum Computing and Quantum Algorithms Teams programs.

\bibliography{research}

\end{document}